\documentclass[12pt]{amsart}

\usepackage[utf8]{inputenc}
\usepackage{amsthm,amsmath,amsfonts,amssymb,amscd,mathrsfs,latexsym}
\usepackage{enumerate,graphicx,cases,extarrows,indentfirst,tabularx}
\usepackage{slashed}

\usepackage[english]{babel}

\newcommand{\C}{\mathbb{C}}
\newcommand{\N}{\mathbb{N}}
\newcommand{\R}{\mathbb{R}}

\newcommand{\Z}{\mathbb{Z}}

\newcommand{\A}{\mathcal{A}}
\newcommand{\cS}{\mathcal{S}}
\newcommand{\cO}{\mathscr{O}}
\newcommand{\td}{\text{d}}
\newcommand{\Spec}{\text{Spec}}
\newcommand{\Res}{\text{Res}}
\newcommand{\pa}{\partial}

\newtheorem{theorem}{Theorem}[section]
\newtheorem{proposition}[theorem]{Proposition}
\newtheorem{lemma}[theorem]{Lemma}
\newtheorem{remark}[theorem]{Remark}
\newtheorem{definition}[theorem]{Definition}
\newtheorem{example}[theorem]{Example}

\title{A note on cohomological mirror symmetry of toric manifolds}
\author{Hao Wen}

\newcommand{\Addresses}{{
  \bigskip
  \footnotesize

  Hao Wen, \textsc{School of Mathematical Sciences and the Key Laboratory of Pure Mathematics and Combinatorics, Nankai University, Tianjin 300071, China}\par\nopagebreak
  \textit{E-mail address} \texttt{wenhao@nankai.edu.cn}

}}

\begin{document}
\maketitle

\begin{abstract}
	In this note we describe a logarithmic version of mirror Landau-Ginzburg model for a semi-projective toric manifold and show the ring of state space of the Landau-Ginzburg model is isomorphic to the $\C$-valued cohomology of the toric manifold.
\end{abstract}




\section{Introduction}

In mirror symmetry beyond the Calabi-Yau case, the mirror of a toric manifold $X_\Sigma$ is given by a Landau-Ginzburg model $(Y,f)$, where $Y=(\C^*)^n$ and $f$ is a Laurent polynomial defined on $Y$.
Related mirror statements are well established during the past years, see for example \cite{G,LLY,Ir,CCIT}.
Among the various results a fundamental one is an isomorphism between the A-model moduli space of $X$ and the B-model moduli space of its mirror $(Y,f)$.
This isomorphism identifies the symplectic geometric structure of $X_\Sigma$ near the large volume limit and the complex geometry structure of $(Y,f)$ near the large complex structure limit.
In particular, it implies the quantum cohomology ring of $X_\Sigma$ is isomorphic to the Jacobi ring of $(Y,f)$.

In this note, we concern the cohomologies \emph{at} the the limit points of the moduli spaces, that's, we will compare the ordinary cohomology ring of a toric manifold $X_\Sigma$ with the state space ring of a limiting version of its mirror.
Motivated by the work of Gross and Siebert, we define such mirror to be a logarithmic version of Landau-Ginzburg model $(Y_\Sigma^\dag,f)$, where $Y_\Sigma^\dag$ is a log Calabi-Yau variety with a log smooth map to the standard log point $0^\dag$, and $f$ is a holomorphic function on $Y_\Sigma^\dag$.
For a fixed semi-projective toric manifold $X_\Sigma$, the pair $(Y_\Sigma^\dag,f)$ can be defined combinatorially from the defining fan $\Sigma$.
In fact, such $(Y_\Sigma^\dag,f)$ is the $q\to 0$ limit of the Hori-Vafa mirror in \cite{HV}. The log sturture and the log smooth map are needed to record the limiting process.

The idea is as follows.
Given a general toric manifold $X_\Sigma$ with defining fan $\Sigma$, it can be viewed as a gluing of affine toric varieties and its cohomology can be computed by a \v{C}ech double complex.
On the other hand, we can construct from combinatoric data of the fan $\Sigma$ an affine space $Y_\Sigma$ and define a differential graded algebra (abbrev. dga)  $L(Y_\Sigma)$ on it.
The point is that we can view $Y_\Sigma$ as a gluing of affine spaces and consider the restriction of $L(Y_\Sigma)$ on each affine piece.
There is a one-to-one correspondence between the pieces on both sides and corresponding pieces have isomorphic cohomologies.
By technically introducing a third dga, we show that the de Rham dga $(\A^\bullet(X_\Sigma),\td)$ is quasi-isomorphic to $L(Y_\Sigma)$.
When $X_\Sigma$ is semi-projective, $Y_\Sigma$ can be equipped with a log structure and there is a natural map from $Y_\Sigma^\dag$ to the standard log point $0^\dagger$. We show that $L(Y_\Sigma)$ is isomorphic to the twisted dga of the Landau-Ginzrbug model $(Y_\Sigma^\dag,f)$.
In this way we explicitly connect the the cohomology ring of $X_\Sigma$ with the state space ring of its mirror Landau-Ginzburg model.

The cohomology of toric manifolds is well-known, and the concept of logarithmic extension of Landau-Ginburg model also appeared previously in \cite{CCIT} (see also \cite{CMW}), hence the novelty of the present note is mainly the explicit connection between the cohomologies of the mirror sides.
The mechanism behind the phenomenon of mirror symmetry is mysterious, attempts to reveal it include homological mirror symmetry (\cite{Ko}), SYZ conjecture (\cite{SYZ}) and Gross-Siebert program (\cite{GS1,GS2}).
Though being elementary, we hope the connection described in this note gives a clue on why we can expect the intricate mirror symmetry phenomenon for toric manifolds.

\section{Spaces associated to the defining fan of a toric manifold}

Any toric variety can be defined by a fan, so we start by specifying the fan structure.
Let $N \cong \Z^n$ be a rank $n$ lattice and $M$ be the dual lattice. Define $N_\R := N\otimes_\Z \R, M_\R := M\otimes _\Z \R$.
Let $\Sigma$ be a smooth fan in $N_\R$, which implies the set of integral generators of rays of every cone $\sigma \in \Sigma$ can be complemented into a $\Z$-basis of $N$. Denote by $\Sigma(k)$ the set of cones of dimension $k$ in $\Sigma$ and let $\Sigma(1) = \{\rho_1,\cdots,\rho_d\}$. Note that a general fan $\Sigma$ may not contain any $n$-dimensional cone.

Let $X_\Sigma$ be the toric variety associated to $\Sigma$, and $X_\sigma$ be the affine toric variety associated to $\sigma \in \Sigma$.
The smoothness of $\Sigma$ implies smoothness of $X_\Sigma$ and each $X_\sigma$ is isomorphic to $\C^k\times (\C^*)^{n-k}$ for $k=\dim \sigma$.
The correspondence $\sigma \mapsto X_\sigma$ is inclusion preserving, that's, when $\sigma$ is a face of $\sigma'$, there is a toric embedding $X_\sigma \hookrightarrow X_{\sigma'}$.
Let $\cS = \{\sigma_1,\cdots,\sigma_s\}$ be a fixed subset of closed cones of $\Sigma$ that cover the support $|\Sigma|$, then $X_\Sigma$ has a covering given by affine toric varieties indexed by $\cS$:
\begin{align} \label{Cover of X}
    X_\Sigma = \bigcup_{\sigma \in \cS} X_\sigma.
\end{align}

There is another affine variety $Y_\Sigma$ naturally associated to $\Sigma$.
Recall that a primitive collection of $\Sigma$ is a subset $P = \{\rho_{i_1}, \cdots,\rho_{i_k}\}$ of $\Sigma(1)$ such that $\rho_{i_1}, \cdots,\rho_{i_k}$do not lie in a common cone of $\Sigma$ while any subset of $P$ do lie in one.
By abuse of notation, let $\rho_i$ denote also the integral generator of the ray $\rho_i$.
Associate to each $\rho_i$ a complex coordinate $z_i$, then each primitive collection corresponds to a monomial $m_P = z_{i_1} \cdots z_{i_k}$.
The Stanley-Reisner ring $R(\Sigma)$ is defined as the quotient
\begin{align*}
    R(\Sigma) := \C[z_1,\cdots,z_d] / \langle \{m_P| P \text{ is a primitive collection.}\} \rangle.
\end{align*}
Geometrically, $R(\Sigma)$ can be viewed as the coordinate ring of an affine variety $Y_\Sigma$ in $\C^d$, with
\begin{align*}
    Y_\Sigma :=\{(z_1,\cdots,z_d)| m_P = 0, \forall P, P \text{ is a primitive collection}\}.
\end{align*}
$Y_\Sigma$ is in fact a union of intersections of coordinate hyperplanes: it can be written as $Y_\Sigma = \bigcup_{\sigma \in \Sigma} Y_\sigma$,
where $Y_\sigma = \{z|z_i=0, \rho_i \notin \sigma\}$.
The correspondence between cones of $\Sigma$ and components of $Y_\Sigma$ is inclusion preserving: $\sigma_i$ is a face of $\sigma_j$ if and only if $Y_{\sigma_i} \subset Y_{\sigma_j}.$
Since the cones in $\cS$ cover $|\Sigma|$, we have
\begin{align}\label{Cover of Y}
	Y_\Sigma = \bigcup_{\sigma \in \cS} Y_\sigma.
\end{align}
	
Note again that since $\Sigma$ may not contain any $n$-dimensional cone, the components $Y_\sigma, \sigma \in \cS$ may have dimension less than $n$.

\section{Presheaves and \v{C}ech complexes}

We will use the language of presheaves on simplicial complex and  the associated \v{C}ech complexes to describe the constructions on both sides of mirror symmetry.

\begin{definition}
    Let $\Delta$ be a simplicial complex. A (covariant) presheaf $\mathcal{V}$ of $\C$-vector space on $\Delta$ consists of the following data:
    \begin{enumerate}
        \item for every subsimplex $\tau$ of $\Delta$, a $\C$-vector space $\mathcal{V}(\tau)$, and
        \item for every inclusion $\tau' \subset \tau$ of simplexes, a morphism of $\C$-vector spaces $r= r_{\tau \tau'}: \mathcal{V}(\tau') \to \mathcal{V}(\tau)$,
    \end{enumerate}
    such that the following conditions are satisfied:
    \begin{enumerate}
        \item $\mathcal{V}(\emptyset) = 0$,
        \item $r_{\tau \tau}$ is identity map on $\mathcal{V}(\tau)$,
        \item if $\tau'' \subset \tau' \subset \tau$ are subsimplexes of $\Delta$, then $r_{\tau \tau''} = r_{\tau \tau'} \circ r_{\tau' \tau''}$.
    \end{enumerate}
\end{definition}

One can associate a \v{C}ech complex to a presheaf $\mathcal{V}$ on $\Delta$.
Let $\Delta_p$ be the set of $p$-simplex of $\Delta$.
Define
\begin{align*}
    \mathcal{C}^p(\Delta, \mathcal{V}) := \bigoplus_{\tau \in \Delta_p} \mathcal{V}(\tau), \quad \mathcal{C}(\Delta, \mathcal{V}) = \bigoplus_p \mathcal{C}^p(\Delta, \mathcal{V}).
\end{align*}
For $\vec{\alpha} = (\alpha(\tau)_{\tau \in \Delta_p}) \in \mathcal{C}^p$, the combinatorial \v{C}ech map $\delta: \mathcal{C}^p(\Delta, \mathcal{V}) \to \mathcal{C}^{p+1}(\Delta, \mathcal{V})$ is given by
\begin{equation*}
    (\delta \alpha)(\tau_{i_0\cdots i_{p+1}}) = \sum_{0\leq k \leq p+1} (-1)^k r_{\tau_{i_0\cdots i_{p+1}} \tau_{i_0\cdots \hat i_k \cdots i_{p+1}}}\alpha(\tau_{i_0\cdots \hat i_k \cdots i_{p+1}}),
\end{equation*}
where $\tau_{i_0\cdots i_p}$ denotes the simplex with vertices $i_0,\cdots,i_p$ and $\hat{i_k}$ implies the index $i_k$ is removed. Here a total ordering of the vertices of $\Delta$ is fixed and when we write  $\tau_{i_0\cdots i_p}$, it is assumed that $i_0 < i_1 < \cdots <i_p$. Clearly $\delta^2 = 0$ and hence $(\mathcal{C}(\Delta, \mathcal{V}), \delta)$ becomes a complex.
Moreover, when each $\mathcal{V}(\tau)$ has a ring structure, there is a natural product:
\begin{align*}
	\mathcal{C}^p(\Delta, \mathcal{V}) \times \mathcal{C}^q(\Delta, \mathcal{V}) \to \mathcal{C}^{p+q}(\Delta, \mathcal{V})
\end{align*}
given by
\begin{align*}
    (\alpha \cup \beta)(\tau_{i_0\cdots i_{p+q}}) := r_{\tau_{i_0\cdots i_{p+q}} \tau_{i_0\cdots  i_p}}\alpha(\tau_{i_0\cdots i_p}) \cdot r_{\tau_{i_0\cdots i_{p+q}} \tau_{i_p\cdots  i_{p+q}}}\beta(\tau_{i_p\cdots i_{p+q}})
\end{align*}
It's well-known that this product structure descends to cohomologies.

In the remaining part of this paper, $\Delta$ will always be the $(s-1)$-simplex with vertices $1,2,\cdots,s$, where $s$ is the number of cones in $\cS$.
To avoid confusion, we will reserve the letter $\sigma$ for cones of $\Sigma$ and letter $\tau$ for subsimplexes of $\Delta$.
Given $\tau \in \Delta$, there is a cone $\sigma_\tau \in \Sigma$ defined as the intersection of $\sigma_i$ with $i \in \tau$.
The map $\tau \mapsto {\sigma_\tau}$ is inclusion reversing: $\tau' \subset \tau$ implies $\sigma_\tau \subset \sigma_{\tau'}$.
Note that it is not necessarily one-to-one, since $\sigma_\tau$ for different $\tau$ may be identical.

Let $\mathcal{V} = \A^k$, by which we mean $\mathcal{V}(\tau)$ is the space $\A^k(X_{\sigma_\tau})$ of smooth complex differential $k$-forms on $X_{\sigma_\tau}$ and $r$ is the natural restriction map.
We have the following complex
\begin{align}\label{smooth exact}
    0 \to \A^k(X_\Sigma) \stackrel{r}{\to} \mathcal{C}^0(\Delta, \A^k) \stackrel{\delta}{\to} \mathcal{C}^1(\Delta, \A^k) \stackrel{\delta}{\to} \mathcal{C}^2(\Delta, \A^k) \stackrel{\delta}{\to} \cdots.
\end{align}
Because of partition of unity, it is exact.

Take $\mathcal{V} = \cO$, by which we mean $\mathcal{V}(\tau)$ is the space $\cO(Y_{\sigma_\tau})$ of algebraic functions on $Y_{\sigma_\tau}$ and $r$ is the restriction map.
Then we have the following complex
\begin{align}\label{algebraic model}
    0 \to \cO(Y_\Sigma) \stackrel{r}{\to} \mathcal{C}^0(\Delta, \cO) \stackrel{\delta}{\to} \mathcal{C}^1(\Delta, \cO) \stackrel{\delta}{\to} \mathcal{C}^2(\Delta, \cO) \stackrel{\delta}{\to} \cdots.
\end{align}
Though partition of unity is no longer available here, we will prove an exactness result that is parallel to (\ref{smooth exact}).
The following observation is of fundamental importance.
Given $Y_\sigma$ and $Y_{\sigma'}$, since each of them is an intersection of coordinate hyperplanes, there exist a canonical projection $p: \C^d \to Y_\sigma$ and an inclusion $i: Y_{\sigma'} \to \C^d$.
Similarly, there exist an inclusion $\hat i: Y_\sigma \cap Y_{\sigma'} \to Y_\sigma$ and a projection $\hat p: Y_{\sigma'} \to Y_\sigma \cap Y_{\sigma'}$.
The observation is that given any function $g$ on $ Y_\sigma$, $$i^* p^* g = \hat p^* \hat i^* g$$ as functions on $ Y_{\sigma'}$.
In this way, a function $g$ on $Y_\sigma$ determines a function on $Y_\Sigma$ in an unambiguous way, which by abuse of notation is still denoted by $g$.

\begin{theorem} \label{algebraic exact}
    The complex (\ref{algebraic model}) is exact.
\end{theorem}

The theorem is not completely new to experts, see for example the Appendix A.2 of \cite{GS1}. 
We will provide an elementary proof in the appendix.

\section{A differential graded algebra that computes $H^\bullet(X_\Sigma)$}

We will construct in this section a dga whose cohomology ring is isomorphic to $H^\bullet(X_\Sigma)$.

\subsection{Local model}
We begin by describing the local model.
Let $\sigma$ be the cone in $\Sigma$ with generating rays $\rho_1,\cdots,\rho_k$.
The de Rham complex $(\mathcal{A}^\bullet(X_\sigma) ,\td)$ computes the cohomology of $X_\sigma$. 
Since $\Sigma$ is smooth, $\rho_1,\cdots,\rho_k$ can be complemented into a $\Z$-basis $\rho_1,\cdots,\rho_n$ of $N$. Let $l_1,\cdots,l_n$ be the dual basis in $M$ with the property that $l_i (\rho_j) = \delta_{ij}$, then the dual cone $\sigma^\vee$ of $\sigma$ has generators $l_1,\cdots,l_k,\pm l_{k+1},\cdots,\pm l_n$ and
\begin{align*}
	X_{\sigma}
	=\, \text{Spec\,} \C[l_1,\cdots,l_k] \times \text{Spec\,} \C[\pm l_{k+1},\cdots, \pm l_n].
\end{align*}
Let $W_\sigma^\bullet$ be the wedge algebra over the $\C$-vector space with generators $l_{k+1},\cdots,\l_n$.
Define
\begin{align*}
	j_\sigma: W_\sigma^1 \ni l \to \frac{\td w^l}{w^l} \in \A^1(X_\sigma)
\end{align*}
where $w^l$ is the monomial in $\C[\pm l_{k+1},\cdots, \pm l_n]$ corresponding to $l$. 
Using the fact that the $\frac{\td w^{l_i}}{w^{l_i}}$'s generate the cohomology ring of $X_\sigma$, one sees that $j_\sigma$ extends to
\begin{align*}
	j_\sigma: (W_\sigma^\bullet, 0) \to (\mathcal{A}^\bullet(X_\sigma) ,\td),
\end{align*}
which is a quasi-isomorphism of differential graded algebras.

Denote by $\Lambda^\bullet$ the exterior algebra of $M_\R \otimes \C$, and let $\xi_1,\cdots,\xi_n$ be a $\Z$-basis of $M$.
We can associate to $Y_\sigma$ a complex
\begin{align*}
	L(Y_\sigma) := (\cO(Y_\sigma) \otimes_\C \Lambda^\bullet, d_L)
\end{align*}
with
\begin{align*}
	d_L(g \otimes \alpha) = \sum_{i=1}^n (\sum_{j=1}^k \xi_i(\rho_j)z_j) g \otimes \frac{\partial}{\partial \xi_i} \alpha,
\end{align*}
where $\frac{\partial}{\partial \xi_i}$ is an odd derivation acting on  $\Lambda^\bullet$ and satisfying $\frac{\partial}{\partial \xi_i}(\xi_j) = \delta_{ij}$.
Because $d_L$ acts as a derivation, $L_\sigma$ has the structure of a dga.
It's easy to see that $d_L$ is independent of the choice of basis $\xi_i's$, so we can choose the basis $l_1,\cdots,l_n$ describe above, under which $d_L$ becomes
\begin{align} \label{explicit d_L}
    d_L(g \otimes \alpha) = \sum_{j=1}^k z_j g \otimes \frac{\partial}{\partial l_j} \alpha.
\end{align}
Here we equip $L(Y_\sigma)$ with a nontrivial gradation: we assign $\deg z_i=2$ and $\deg \xi_i=1$ so that $d_L$ acts as a differential of degree $1$.
Such a grading may seem strange, however, as we will see in the next section, it's compatible with definition of $\{f,-\}$ on logarithmic polyvector fields and then finally compatible with the cohomological degree of $H(X_\Sigma)$.
Let $U_\sigma^\bullet$ be the wedge algebra over $\C$ with generators $l_1,\cdots,l_k$, then we have $\Lambda^\bullet = W_\sigma^\bullet \otimes U_\sigma^\bullet$ and $L(Y_\sigma)$ is the tensor product of $(W_\sigma^\bullet, 0)$ and $(\cO(Y_\sigma) \otimes U_\sigma^\bullet, d_L)$.
The latter is a Koszul complex. 
As $z_1,\cdots,z_k$ forms a regular sequence in $\cO(Y_\sigma)$, $(\cO(Y_\sigma) \otimes U_\sigma^\bullet, d_L)$ has cohomology $\C \cdot 1$.
One concludes that $L(Y_\sigma)$ has cohomology $W_\sigma^\bullet$ and the natural inclusion
\begin{align*}
	k_\sigma: (W_\sigma^\bullet, 0) \to L(Y_\sigma)
\end{align*}
is a quasi-isomorphism of dga.

\subsection{\v{C}ech model}

Next we assemble the constructions above to maps between \v{C}ech complexes, for which we need to define compatible restriction maps between local models.

Let $\sigma \subset \sigma'$ be cones in $\Sigma$, we have $X_\sigma \subset X_{\sigma'}$ and $Y_\sigma \subset Y_{\sigma'}$.
For $(\mathcal{A}^\bullet(X_\sigma) ,\td)$, the restriction map 
\begin{align*}
	r_{\sigma \sigma'}^\A: \mathcal{A}^\bullet(X_{\sigma'}) \to \mathcal{A}^\bullet(X_\sigma)
\end{align*}
is compatible with de Rham differential $\td$.
As to $L(Y_\sigma)$, we define
\begin{align*}
	r_{\sigma \sigma'}^L: L(Y_{\sigma'}) \to L(Y_\sigma),\quad g \otimes \alpha \mapsto g|_{Y_\sigma} \otimes \alpha.
\end{align*}
By the explicit expression (\ref{explicit d_L}) of $d_L$ one sees $r_{\sigma \sigma'}^L$ is indeed a chain map.
For $(W_\sigma,0)$, one observes that $W_\sigma^1$ is exactly the annihilator of $\sigma$ in $M_\R \otimes \C$.
When $\sigma \subset \sigma'$, we have $W_{\sigma'}^1 \subset W_\sigma^1$. So we can define 
\begin{align*}
	r_{\sigma \sigma'}^W: W_{\sigma'}^\bullet \to W_\sigma^\bullet
\end{align*}
as the inclusion map.
It's as expected that the maps $j_\sigma$ and $k_\sigma$ defined above are compatible with the various restriction maps:
the identities
\begin{align*}
	r_{\sigma\sigma'}^A \circ j_{\sigma'} = j_\sigma \circ r_{\sigma\sigma'}^W: (W_{\sigma'}^\bullet, 0) \to (\mathcal{A}^\bullet(X_\sigma),\td)
\end{align*}
and
\begin{align*}
	r_{\sigma\sigma'}^L \circ k_{\sigma'} = k_\sigma \circ r_{\sigma\sigma'}^W: (W_{\sigma'}^\bullet, 0) \to L(Y_\sigma)
\end{align*}
both follow from the explicit constructions above.

Define $\A:= \oplus \A^{k,p}$ with $\A^{k,p} := \oplus_{\tau \in \Delta_p} \A^k(X_{\sigma_\tau})$.
Using that restriction maps are compatible with the differential map, we obtain a \v{C}ech double complex $(\A,\delta,\td)$.
By the exactness of (\ref{smooth exact}), we have the following quasi-isomorphism
\begin{align}
	(\A^\bullet(X_\Sigma),\td) \hookrightarrow (\A,\delta,\td).
\end{align}

Similar to the definition of $L(Y_\sigma)$, we can define a dga on $Y_\Sigma$ as
\begin{align*}
	L(Y_\Sigma) :=& (\cO(Y_\Sigma)\otimes \Lambda^\bullet, d_L), \\ d_L(g \otimes \alpha) =& \sum_{i=1}^n (\sum_{j=1}^d \xi_i(\rho_j)z_j) g \otimes \frac{\partial}{\partial \xi_i} \alpha.
\end{align*}
We can also define a chain map
\begin{align*}
	r_\sigma: L(Y_\Sigma) \to L(Y_\sigma),\quad g \otimes \alpha \mapsto g|_{Y_\sigma} \otimes \alpha.
\end{align*}
Write $L:= \oplus L^{k,p}$ with $L^{k,p} := \oplus_{\tau \in \Delta_p} \cO(Y_{\sigma_\tau}) \otimes \Lambda^k$, we obtain a \v{C}ech double complex $(L,\delta,d_L)$.
By Theorem \ref{algebraic exact}, the complex
\begin{align*}
	0 \to \cO(Y_\Sigma) \otimes_\C \Lambda^k \stackrel{r}{\to} \mathcal{C}^0(\Delta, \cO \otimes_\C \Lambda^k) \stackrel{\delta}{\to} \mathcal{C}^1(\Delta, \cO \otimes_\C \Lambda^k) \stackrel{\delta}{\to}  \cdots,
\end{align*}
is exact for $\forall k$, so we get another quasi-isomorphism
\begin{align} \label{Alg quasi-isomorphism}
	L(Y_\Sigma) \hookrightarrow (L, \delta, d_L).
\end{align}

\begin{proposition} \label{quasi-isomorphism}
	$L(Y_\Sigma)$ and $(\A(X_\Sigma),\td)$ are quasi-isomorphic dga's, hence  $L(Y_\Sigma)$ computes the $\C$-valued cohomology ring of $X_\Sigma$.
\end{proposition}
\begin{proof}
	By the discussions above, it suffices to show the double complexes $(\A,\delta,\td)$ and $(L, \delta, d_L)$ are quasi-isomorphic as dga's.
	To show this, we define a third double complex $(W,\delta,0)$ with
	\begin{align*}
		W:= \oplus W^{k,p},\quad W^{k,p}:= \oplus_{\tau \in \Delta_p} W^k_{\sigma_\tau}.
	\end{align*}
	Since $j_\sigma$ and $k_\sigma$ are compatible with restriction maps, they can be assembled into maps of double complexes
	\begin{align*}
		j: (W,\delta,0) \to (\A,\delta,\td)
	\end{align*}
	and 
	\begin{align*}
		k: (W,\delta,0) \to (L, \delta, d_L).
	\end{align*}
	Since $j_\sigma$'s and $k_\sigma$'s are all quasi-isomorphisms, both $j$ and $k$ induces isomorphic vertical cohomology groups. By the theory of spectral sequence, both $j$ and $k$ are quasi-isomorphisms and hence $(\A,\delta,\td)$ and $(L, \delta, d_L)$ are quasi-isomorphic double complexes.
	Moreover, it's obvious that the all various chain maps preserve the product structure on \v{C}ech complexes, so $(\A,\delta,\td)$ and $(L, \delta, d_L)$ are indeed quasi-isomorphic as dga's.
\end{proof}

\section{Landau-Ginzburg mirror for semiprojective toric manifolds}

The results in the previous sections are valid for all toric manifolds. 
In this section we will interpret $L(Y_\Sigma)$ as the twisted complex of a Landau-Ginzburg model, under the additional assumption that $X_\Sigma$ is semi-projective.
By Proposition 7.2.9 of \cite{CLS}, $X_\Sigma$ being semi-projective is equivalent to saying $X_\Sigma$ is the toric variety of a full dimensional lattice polyhedron $P \subset M_\R$. 
As the normal fan of $P$, $\Sigma$ has full dimensional convex support.
We can define a strictly convex $\Sigma$-piecewise linear function with integral slopes on $|\Sigma|$ as follows
\begin{align*}
    \varphi(n) := - \text{inf}\,\{(m,n)| m \in P\}.
\end{align*}
Consider the subset of $N_{\R}\oplus \R$ defined as
\begin{equation*}
    Q =\{(n,r) | n \in |\Sigma|, r\geq \varphi(n)\},
\end{equation*}
this is a strictly convex rational polyhedral cone and hence defines an affine toric variety
\begin{equation*}
    \mathscr{Y} = \text{Spec}(\C[Q \cap (N\oplus \Z)]).
\end{equation*}

Let $Q_v$ be the set of vertical components of the boundary $\partial Q$ of $Q$.
Such component will  appear only when $|\Sigma|$ is a proper subset of $N_\R$.
They define a divisor $\mathcal{D}$ in $\mathscr{Y}$.
Denote by $Q_h$ the set of horizontal components of $\partial Q$, it defines another divisor $Y$ in $\mathscr{Y}$.
To identify $Y$, consider the natural inclusion $\N \hookrightarrow Q \cap (N\oplus \Z)$ given by $k \mapsto (0,k)$.
This induces a map $\pi: \mathscr{Y} \to \C = \text{Spec}(\C[\N])$, and $Y$ is exactly the fiber over $0$.
In fact, $Y$ is isomorphic to $Y_\Sigma$.
To see this, we define a monoid structure on $(Q_h \cap (N\otimes \Z)) \cup \{\infty\}$ as follows:
\begin{eqnarray}
    p+q= \left\{
    \begin{aligned}
        p+q \quad& \text{if }p,q,p+q \in Q_h; \\
        \infty \quad& \text{otherwise}.
    \end{aligned}
    \right.
\end{eqnarray}
By formally putting $z^\infty=0$, we have $Y=\Spec\,\C[(Q_h \cap (N\otimes \Z)) \cup \{\infty\}]$.
Let $p: Q_h\cap (N\otimes \Z) \to \Sigma \cap N$ be the canonical projection induced by $p: N\oplus \Z \to N$, which is one-to-one and induces a monoid structure on $\overline \Sigma := (\Sigma \cap N) \cup \{\infty\}$.
We have $Y \cong \text{Spec\,}\C[\overline \Sigma]$ and hence
\begin{equation*}
    Y\cong Y_\Sigma
\end{equation*}
by the definition of $Y_\Sigma$.

To sum up, we get a flat morphism $\pi: \mathcal{Y} \to \C$ together with a reduced divisor $\mathcal{D} \subset \mathcal{Y}$.
It is easy to see they form a toric degeneration of CY-pairs over $\C$ in the sense of Definition 1.8 in \cite{GS}.
In particular, if we write $D := \mathcal{D} \cap Y_\Sigma$, then the pair $(Y_\Sigma,D)$ forms a totally degenerate CY-pair.

In the following we need some knowledge of log geometry, which can be found in \cite{K}.
Equip $\mathscr{Y}$ with the log structure induced by the inclusion $Y_\Sigma \cup \mathcal{D} \hookrightarrow \mathscr{Y}$ and equip $\C$ with the log structure induced by the inclusion $0 \hookrightarrow \C$.
Using the superscript $\dag$ to emphasize that a variety is equipped with a given log structure, we get log smooth maps
\begin{equation*}
    \Phi:\mathscr{Y}^\dag \to \C^\dag
\end{equation*}
and
\begin{equation*}
    \phi: Y_\Sigma^{\dag} \to 0^\dag,
\end{equation*}
where the latter is obtained by pulling back the log structure on corresponding ambient spaces.

\begin{lemma}
    The map $\phi: Y^{\dag} \to 0^\dag$ is log smooth.
\end{lemma}
\begin{proof}
It follows from Theorem 4.1 of \cite{K}.
\end{proof}

Since $\phi$ is log smooth, the sheaves $\Omega^\bullet_{Y_\Sigma^\dag / 0^\dag}$ and $\Theta^\bullet_{Y_\Sigma^\dag / 0^\dag}$ of relative log differentials and log derivations are locally free on $Y_\Sigma$.
These sheaves can be described explicitly as follows.
By the semi-projectivity of $X_\Sigma$, $|\Sigma|$ is covered by $n$-dimensional cones. Let $\sigma_M$ be one of them.
By smoothness of $\Sigma$, we can assume without loss of generality that $\rho_1=e_1,\cdots,\rho_n=e_n$ is the set of generating rays of $\sigma_M$.

Assume that when $n+1 \leq l \leq d$,
\begin{equation*}
    \rho_l - \sum_{i=1}^n a_{il} \rho_i = 0.
\end{equation*}
By construction of $\pi: \mathscr{Y} \to \C$, we have
\begin{equation*}
    z_l \cdot \prod_{i=1}^n z_i^{-a_{il}} = t^m
\end{equation*}
for some $m \in \Z_+$, where $t$ is the generator of $\Spec \, \C[\N] \cong \C[t]$.
Hence
\begin{equation} \label{relation of log differential}
    \frac{\td z_l}{z_l} - \sum_{i=1}^n a_{il}\cdot \frac{\td z_i}{z_i}= m \frac{\td t}{t}, \quad \forall n+1 \leq l \leq d
\end{equation}
and $\Omega^1_{Y_\Sigma^\dag / 0^\dag}$ is defined as the $\cO_{Y_\Sigma}$-module $\oplus_{i=1}^d \cO_{Y_\Sigma} \frac{\td z_i}{z_i}$ modulo the relations
\begin{equation} \label{relation of log differential}
    \frac{\td z_l}{z_l} - \sum_{i=1}^n a_{il}\cdot \frac{\td z_i}{z_i}= 0, \quad \forall n+1 \leq l \leq d.
\end{equation}
We define $\Omega^k_{Y_\Sigma^\dag / 0^\dag} := \wedge^k \Omega^1_{Y_\Sigma^\dag / 0^\dag}$.
On any $n$-dimensional component $Y_\sigma$ of $Y_\Sigma$, any $\alpha \in \Omega^k_{Y_\Sigma^\dag / 0^\dag}$ has a unique representative $\alpha|_{Y_\sigma}$ given by an algebraic $k$-form with at most log poles on $Y_\sigma \cap \text{Sing}\, Y_\Sigma$. 
If $Y_{\sigma_i}$ and $Y_{\sigma_j}$ has $(n-1)$-dimensional intersection $D_{ij}$, then
\begin{equation*}
    \Res_{D_{ij}} \alpha|_{Y_{\sigma_i}} + \Res_{D_{ij}} \alpha|_{Y_{\sigma_j}} = 0,
\end{equation*}
where $\Res$ denotes the Poincar\'e residue map.
Write $\theta_i := z_i\frac{\pa}{\pa z_i}$, then $\Theta^1_{Y_\Sigma^\dag / 0^\dag}$ is the $\cO_{Y_\Sigma}$-submodule
\begin{equation} \label{relation of log derivation}
    \{\sum_{i=1}^d g_i \theta_i \big| g_l = \sum_{i=1}^n g_i a_{il}, \forall n+1 \leq l \leq d\}
\end{equation}
of $\oplus_{i=1}^d \cO_{Y_\Sigma}\theta_i$ and $\Theta^k_{Y_\Sigma^\dag / 0^\dag} = \wedge^k \Theta^1_{Y_\Sigma^\dag / 0^\dag}$.
We will call global sections of $\Theta^\bullet_{Y_\Sigma^\dag / 0^\dag}$ logarithmic polyvector fields.
More explicitly, $\Theta^1_{Y_\Sigma^\dag / 0^\dag}$ is a rank $n$ free $\cO_{Y_\Sigma}$-module with generators $\tilde\theta_i, 1\leq i \leq n$, where $\tilde\theta_i = \theta_i + \sum_{l=n+1}^d a_{il} \theta_l$.
The map $\tilde \theta_i \mapsto \theta_i$ gives an isomorphism
\begin{equation}\label{identification}
\Theta^1_{Y_\Sigma^\dag / 0^\dag} \cong \bigoplus_{1\leq i \leq n} \cO_{Y_\Sigma} \cdot \theta_i.
\end{equation}

There is a global nowhere vanishing log differential $n$-form $\Omega$ such that it can be written as $\wedge_{i=1}^n \frac{\td z_i}{z_i}$ on $Y_{\sigma_M}$, hence $Y_\Sigma^\dag$ is in fact log Calabi-Yau.
Denote $f = z_1 +z_2 + \cdots +z_d$, then we get a logarithmic version of Landau-Ginzburg model
\begin{align*}
	(Y_\Sigma^\dag, f)
\end{align*}
and it will be regarded as the mirror of $X_\Sigma$.

\begin{remark}
	Such mirror of $X_\Sigma$ should be viewed as (generalization of) a limiting version of Hori-Vafa mirror in \cite{HV}.
	Take $X_\Sigma = \mathbb{P}^1$ for instance.
	The Hori-Vafa mirror of $\mathbb{P}^1$ is given by a family of Landau-Ginzburg models parametrized over $q \in \C^*$:
	\begin{align*}
		(\C^*, f=z+ \frac{q}{z}).
	\end{align*}
	The $q=0$ limit corresponds to the large complex structure limit.
	Rewrite the family as
	\begin{align*}
		(Y_q:=\{z_1z_2 = q\} \subset \C^2, f=z_1 +z_2),
	\end{align*}
	then when $q=0$, $Y_q$ becomes $\{z_1z_2 = 0\}$. This is exactly the variety $Y_\Sigma$ given by the defining fan of $\mathbb{P}^1$.
	The log structure on $Y_\Sigma$ and the log smooth map to $0^\dagger$ record the information that $Y_\Sigma$ is the central fiber of the family $\C^2 \to \C, (z_1,z_2) \mapsto z_1 z_2$.
\end{remark}

When written as a log differential, $\td f = \sum_{i=1}^d z_i \frac{\td z_i}{z_i}$ defines a map $\td f \wedge:\Omega^\bullet_{Y_\Sigma^\dag / 0^\dag} \to \Omega^{\bullet+1}_{Y_\Sigma^\dag / 0^\dag}$.
As usual, contraction with $\Omega$ defines a map
\begin{equation*}
    \Theta^k_{Y_\Sigma^\dag / 0^\dag} \to \Omega^{n-k}_{Y_\Sigma^\dag / 0^\dag}, \quad \theta \mapsto \theta \vdash \Omega.
\end{equation*}
It then induces from $\td f \wedge$ a map
$$
\{f,-\}:\Theta^\bullet_{Y_\Sigma^\dag / 0^\dag} \to \Theta^{\bullet-1}_{Y_\Sigma^\dag / 0^\dag}.
$$
$(\Theta^\bullet_{Y_\Sigma^\dag / 0^\dag}, \{f,-\})$ is a differential graded algebra since $\{f,-\}$ is a derivation.
Under the identification (\ref{identification}), we can write
$$
\{f,-\} = \sum_{i=1}^n (z_i + \sum_{l=n+1}^d a_{il}z_i) \frac{\pa}{\pa \theta_i},
$$
where $\frac{\pa}{\pa \theta_i}$ is an odd derivation satisfying $\frac{\pa}{\pa \theta_i} (\theta_j) = \delta_{ij}$.
We equip $(\Theta^\bullet_{Y_\Sigma^\dag / 0^\dag}, \{f,-\})$ with a grading structure by specifying $\deg z_i = 2$ and $\deg \theta_i = 1$ so that $\{f,-\}$ has degree $1$.

\begin{theorem}
When $X_\Sigma$ is a semi-projective toric manifold, then
\begin{enumerate}[1)]
    \item The dga $(\Theta^\bullet_{Y_\Sigma^\dag / 0^\dag} (Y_\Sigma), \{f,-\})$ is isomorphic to $L(Y_\Sigma)$;
    \item There is  a ring isomorphism
    \begin{align*}
    	H(\Theta^\bullet_{Y_\Sigma^\dag / 0^\dag} (Y_\Sigma), \{f,-\}) \cong H(X_\Sigma).
    \end{align*}
\end{enumerate}
\end{theorem}
\begin{proof}
    Statement 2) follows from statement 1) and Proposition \ref{quasi-isomorphism}.
    For statement 1), one needs only to choose a basis $\xi_1,\cdots,\xi_n$ of $M$ that satisfies $\xi_i(\rho_j)=\delta_{ij}, \forall 1\leq j \leq n$, then the map $\theta_i \mapsto\xi_i, i=1,\cdots,n$ gives the desired isomorphism.
\end{proof}

\begin{example}
	When $X_\Sigma$ is projective, $\Theta^\bullet_{Y_\Sigma^\dag / 0^\dag} (Y_\Sigma)$ is the exterior algebra over the Stanley-Reisner ring $R(\Sigma)$ with generators $\theta_i = z_i\frac{\pa}{\pa z_i}, 1\leq i \leq n$ and $\{f,-\} = \sum_{i=1}^n (z_i + \sum_{l=n+1}^d a_{il}z_i) \frac{\pa}{\pa \theta_i}$.
	It's straight-forward to verify that $z_i + \sum_{l=n+1}^d a_{il}z_i, 1\leq i \leq n$ form a regular sequence in $R(\Sigma)$, hence 
	\begin{align*}
		H(\Theta^\bullet_{Y_\Sigma^\dag / 0^\dag} (Y_\Sigma), \{f,-\}) \cong R(\Sigma)/\langle z_i + \sum_{l=n+1}^d a_{il}z_i\rangle_{1\leq i \leq n}.
	\end{align*}
	The right hand side is exactly the well-known cohomology ring of $X_\Sigma$.
	The reason for the above gradation is now clear: it induces a grading on $H(\Theta^\bullet_{Y_\Sigma^\dag / 0^\dag} (Y_\Sigma), \{f,-\})$ that is compatible with the usual cohomological grading of $H(X_\Sigma)$.
\end{example}

The above example is included in \cite{CMW}. In this note $X_\Sigma$ is allowed to be only semi-projective.

\begin{example}
    Let $\Sigma$ be the $n$-dimensional cone generated by the basis $e_1,\cdots,e_n$ of $N$.
    On the A-side, $X_\Sigma = \C^n$, $H(\C^n,\C) \cong \C$. 
    On the mirror side, $Y_\Sigma = \Spec\, \C[z_1,\cdots,z_n] \cong \C^n$, $D$ is the divisor defined by $z_1\cdots z_n = 0$.
    In this case, $\Theta^\bullet:=\Theta^\bullet_{Y_\Sigma^\dag / 0^\dag} (Y_\Sigma)$ is the exterior algebra over $\C[z_1,\cdots,z_n]$ with generators $\theta_i = z_i\frac{\pa}{\pa z_i}, 1\leq i \leq n$.
    The twisted dga then has the form
    \begin{align*}
        0\to \Theta^n \stackrel{\sum z_i \frac{\pa}{\pa \theta_i}}{\longrightarrow} \Theta^{n-1} \stackrel{\sum z_i \frac{\pa}{\pa \theta_i}}{\longrightarrow} \cdots \stackrel{\sum z_i \frac{\pa}{\pa \theta_i}}{\longrightarrow} \Theta^1 \stackrel{\sum z_i \frac{\pa}{\pa \theta_i}}{\longrightarrow} \C[z_1,\cdots,z_n] \to 0.
    \end{align*}
    Since $z_1,\cdots,z_n$ form a regular sequence in $\C[z_1,\cdots,z_n]$, we conclude that $H(\Theta^\bullet_{Y_\Sigma^\dag / 0^\dag} (Y_\Sigma), \{f,-\}) \cong \C$.
    The cohomolgy on both sides concentrate on degree $0$.
\end{example}

The above example serves as a local model, we may consider the construction for general semi-projective toric manifolds as a gluing of such local pieces.

\begin{example}
	Let $\Sigma$ be the convex fan in $\Z^2$ with ray generators $\rho_1=e_1, \rho_2 = e_2,\rho_3=-e_2$. We have $X_\Sigma \cong \C \times \mathbb{P}^1$ and $H(X_\Sigma) \cong H(\mathbb{P}^1,\C) \cong \C[z]/\langle z^2 \rangle$. 
	For Landau-Ginzburg side, $\Theta^\bullet_{Y_\Sigma^\dag / 0^\dag} (Y_\Sigma)$ is the wedge algebra over $R(\Sigma) = \C[z_1,z_2,z_3]/\langle z_2 z_3 \rangle$ with generators $\theta_1, \theta_2$ and $\{f,-\} = z_1 \frac{\pa}{\pa \theta_1} + (z_2-z_3) \frac{\pa}{\pa \theta_2}$. Similarly, we can verify that $z_1,z_2-z_3$ is a regular sequence in $R(\Sigma)$, so
	\begin{align*}
		H(\Theta^\bullet_{Y_\Sigma^\dag / 0^\dag} (Y_\Sigma), \{f,-\}) \cong& \C[z_1,z_2,z_3]/\langle z_2z_3,z_1,z_2-z_3 \rangle \cong \C[z]/\langle z^2 \rangle \\ \cong& H(X_\Sigma)
	\end{align*}
	as graded algebras.
\end{example}

\appendix
\section{Proof of Theorem \ref{algebraic exact}}

To simplify the notation, we will denote by $Y_{i_0\cdots i_p}$ the intersection $Y_{\sigma_{i_0}}\cap \cdots \cap Y_{\sigma_{i_p}}$.
For $\tau \subset \Delta$, we will write $Y_\tau$ for $Y_{\sigma_\tau}$.

\begin{lemma} \label{p=0}
	The $p=0$ joint of complex (\ref{algebraic model}) is exact.
\end{lemma}
\begin{proof}
	Assume $\vec{g} = (g_1,\cdots,g_s), g_i \in \cO(Y_i)$ satisfies $\delta (\vec{g}) = 0$, one has to show there is a global $g\in \cO(Y_\Sigma)$ which restrict to $g_i$ on each $Y_i$.
	Closeness of $\vec{g}$ implies $g_i = g_j$ on $Y_{ij}$, it's clear then $\vec{g}$ determines a unique function $g_{i_0\cdots i_p}$ on each $Y_{i_0\cdots i_p}$.
	Define
	\begin{align} \label{global function}
		g = \sum_{p=0}^s \, \sum_{1\leq i_0<\cdots <i_p \leq s} (-1)^p g_{i_0\cdots i_p},
	\end{align}
	it is sufficient to show $g|_{Y_i} = g_i, 1 \leq i \leq s$.
	Take $i=1$ for instance.
	Fix a nonempty subset $\mathcal{T}$ of $\{2,\cdots,s\}$, and let $\tau$ and $\tau'$ be the subsimplexes in $\Delta$ with vertices in $\mathcal{T}$ and $\mathcal{T} \cup \{1\}$ respectively.
	Since the intersection of $Y_\tau$ with $Y_1$ is exactly $Y_{\tau'}$ and $g_\tau$ restricts to $g_{\tau'}$ on $Y_{\tau'}$, the restriction of $(-1)^{|\mathcal{T}|}g_{\tau'}$ and $(-1)^{|\mathcal{T}|-1}g_\tau$ on $Y_1$ cancel out.
	Hence when restricted to $Y_1$ and after cancellation, the only term on the R. H. S. of (\ref{global function}) that remains is $g_1$, as expected.
\end{proof}

\begin{lemma} \label{parametrized solution}
	Assume matrices $A \in \C^{p\times q}, B\in \C^{q\times r}$ satisfies $\{\vec{a} \in \C^{q \times 1}|A \vec{a} = 0\} \subset \{B\vec{b}|\vec{b} \in \C^{r\times 1}\}$.
	Let $\vec{a}(z) := (a_1(z),\cdots,a_q(z))^T$ be a family of vectors that depend algebraically on parameters $z = (z_1,\cdots,z_k)$ and that for any $z$, $A \vec{a}(z)=0$.
	Then we can find vectors $\vec{b}(z)$ depending algebraically on $z$ such that
	\begin{align*}
		B \vec{b}(z) = \vec{a}(z).
	\end{align*}
	Moreover, we can require $\vec{b}(z) = 0$ if $\vec{a}(z) = 0$.
\end{lemma}
\begin{proof}
	This is an easy exercise in linear algebra.
\end{proof}

\begin{lemma} \label{p>0}
	The $p>0$ joint of complex (\ref{algebraic model}) is exact.
\end{lemma}
\begin{proof}
	Assume $\vec{g} = (g_\tau)_{\tau \in \Delta_p}$ satisfies $\delta(\vec{g}) = 0$. This means to each $p$-simplex $\tau$ is associated a function $g_\tau$ and the functions satisfy a set of linear compatibility conditions.
	Consider the complex $(\mathcal{C}(\Delta,\C), \delta)$ with $\mathcal{V}(\tau) = \C$ for every $\tau \subset \Delta$, it's easy to see this complex computes the $\C$-valued singular cohomology of $\Delta$. Hence for $p>0$, the cohomologies vanish.
	
	For each $z \in Y_\Delta$, $\vec{g}(z) = (g_\tau(z))_{\tau \in \Delta_p}$ determines a $p$-cycle in $(\mathcal{C}(\Delta,\C), \delta)$, hence there is a $(p-1)$-chain $\vec{h}(z) \in (\mathcal{C}(\Delta,\C),\delta)$ such that $\delta(\vec{h}(z)) = \vec{g}(z)$.
	By Lemma \ref{parametrized solution}, one can require $\vec{h}(z)$ to depend algebraically on $z \in Y_\Delta$ to get a $(p-1)$-chain $\vec{h} \in (\mathcal{C}(\Delta,\cO(Y_\Delta),\delta)$.
	The canonical projection maps $Y_\tau \to Y_\Delta, \tau \in \Delta_{p-1}$ lift $\vec{h}$ to a $(p-1)$-chain in $(\mathcal{C}^{p-1}(\Delta,\cO),\delta)$, which will be again denoted by $\vec{h}$.
	Now each component of the cycle $\vec{g}^{(1)} := \vec{g}-\delta \vec{h}$ will restrict to $0$ on $Y_\Delta$.
	
	Fix an $(s-2)$-simplex $\omega$, we have $g^{(1)}_{\tau}|_{Y_{\omega}} = 0$ if $\tau \subsetneq \omega$ and $(\vec{g}^{(1)}_\tau|_{Y_{\omega}})_{\tau \subset \omega}$ form a $p$-chain $\vec{g}^{(1),\omega}$ in $(\mathcal{C}(\omega, \cO(Y_\omega)),\delta)$.
	By Lemma \ref{parametrized solution}, there is a $(p-1)$-chain $\vec{h}^{(1),\omega}$ in $(\mathcal{C}(\omega, \cO(Y_\omega)),\delta)$ satisfying $\delta(\vec{h}^{(1),\omega})= \vec{g}^{(1),\omega}$ and each component of which will restrict to $0$ on $Y_\Delta$.
	Define $\vec{h}^{(1),\omega}_{\rho} = 0$ if $\rho \subsetneq \omega$, lift them to a $(p-1)$-chain in $(\mathcal{C}(\Delta,\cO),\delta)$ and write $\vec{g}^{(2)} := \vec{g}^{(1)} - \delta(\sum_{\omega \in \Delta_{s-2}} \vec{h}^{(1),\omega})$, we find each component of $\vec{g}^{(2)}$ will restrict to $0$ on each $Y_\omega, \omega \in \Delta_{s-2}$.
	By induction, we can find a $p$-cycle $\vec{g}^{(s-p-1)}$ in $(\mathcal{C}(\Delta,\cO),\delta)$ such that $\vec{g}^{(s-p-1)}$ is cohomologous to $\vec{g}$ and each component of $\vec{g}^{(s-p-1)}$ will restrict to $0$ on each $Y_\omega, \omega \in \Delta_{p+1}$.
	
	Now for each $p$-simplex $\tau$, $\vec{g}^{(s-p-1)}_{\tau}$ can itself be viewed as a $p$-cycle in $(\mathcal{C}(\Delta,\cO),\delta)$.
	Choose a $(p-1)$-simplex $\omega_\tau$ such that $\omega_\tau \subset \tau$, and lift $\vec{g}^{(s-p-1)}_{\tau}$ to a function $\vec{h}^{(s-p-1)}_\tau$ on $Y_{\omega_\tau}$ by the canonical projection. Clearly, $\delta(\vec{h}^{(s-p-1)}_\tau) = \pm \vec{g}^{(s-p-1)}_\tau$ and $\vec{h}^{(s-p-1)}_\tau|_{Y_\omega} = 0,\forall \omega \in \Delta_{p+1}$.
	Then $\vec{g}^{(s-p-1)} = \delta(\sum_{\tau \in \Delta_p} \pm \vec{h}^{(s-p-1)}_\tau)$.
	Take everything into consideration, $\vec{g}$ is $\delta$-exact.
\end{proof}

Theorem \ref{algebraic exact} follows from  Lemma \ref{p=0} and Lemma \ref{p>0}.

\Addresses

\end{document}